\documentclass[aps,a4paper,reprint,eprint,onecolumn,superscriptaddress,longbibliography]{revtex4-2}
\usepackage[T1]{fontenc}
\usepackage[utf8]{inputenc}
\usepackage[english]{babel}
\usepackage{indentfirst}
\usepackage{lmodern}
\usepackage[intlimits]{amsmath}
\usepackage{bm}
\usepackage{amsthm, amssymb, amsfonts}
\usepackage{braket}
\usepackage{todonotes}
\usepackage{chngcntr}
\usepackage{csquotes}

\newcommand{\R}{\mathbb{R}}

\newcommand{\ee}{\mathrm{e}}
\newcommand{\rd}{\mathrm{d}}
\newcommand{\bmr}{{\bm r}}
\newcommand{\eps}{\varepsilon}

\newtheorem*{theorem*}{Theorem}
\newtheorem{theorem}{Theorem}

\newtheorem*{cor*}{Corollary}
\newtheorem{lemma}[theorem]{Lemma}
\newtheorem{definition}{Definition}

\newtheorem*{ass*}{Assumption}

\usepackage{graphicx,grffile,subfigure}
\graphicspath{{figures/}}
\newlength\figwidth
\usepackage{url}

\usepackage[capitalize]{cleveref}
\crefrangelabelformat{equation}{(#3#1#4)$-$(#5#2#6)}

\begin{document}

\title{Scaling law for the size dependence of a {{finite-range}} quantum gas}

\newcommand\FUBaffiliation{\affiliation{Freie Universität Berlin, Institute of Mathematics, Arnimallee 6, 14195 Berlin, Germany}}
\newcommand\BTUaffiliation{\affiliation{Brandenburgische Technische Universität Cottbus-Senftenberg, Institute of Mathematics, Konrad-Wachsmann-Allee 1, 03046 Cottbus, Germany}}
\author{Luigi Delle Site}
\email{luigi.dellesite@fu-berlin.de}
\FUBaffiliation
\author{Carsten Hartmann}
\email{carsten.hartmann@b-tu.de}
\BTUaffiliation

\begin{abstract}
In a recent work [Reible {\it et al.} Phys. Rev. Res. {\bf 5}, 023156, 2023], it has been shown that the mean particle-particle interaction across an ideal surface that divides a system into two parts, can be employed to estimate the size dependence for the thermodynamic accuracy of the system. In this work we propose its application to systems with {finite} range interactions that models a dense quantum gases and derive an approximate size-dependence scaling law. In addition, we show that the application of the criterion is equivalent to the determination of a free energy response to a perturbation. The latter result confirms the complementarity of the criterion to other estimates of finite-size effects based on direct simulations and empirical structure or energy convergence criteria.
\end{abstract}

\maketitle

\section{Introduction}
Quantum gases represent a highly active field of research due to their unusual physical properties and potential applications in science and technology, one such example being ultracold quantum gases that can be used as quantum simulators (see, e.g., Ref.\citenum{cold1} and the references therein). To properly investigate the complexity of such systems often requires an interplay between experimental and computational techniques\cite{revmodphys}. In doing so, the physical consistency of the theoretical and computational models is a key issue that can be controlled by the system size. In computer simulation, there is always a trade-off between the physical and empirical consistency---that requires a fairly  large number of particles---and the computational cost associated with simulating a large system as many algorithms scale at least polynomially (sometimes even exponentially) in the number of degrees of freedom. 
Fortunately, many interesting systems are rather dilute, so few particles in a relatively large box are sufficient to properly reproduce thermodynamic (bulk) properties \cite{qmc}.

The situation is different for dense quantum gases that are {\it per se} a field of large interest across various disciplines, from chemistry to mathematics \cite{dense1,dense2,dense3,dense4,dense5}, or even astrophysics\cite{astro}. Here using few particles when simulating bulk properties may lead to inconsistent results, and thus \textit{a priori} estimates to quantify finite size effects are crucial to assess the accuracy and reliability of simulation results\cite{ceper}.

The aim of this paper is to provide such estimates for dense systems with {finite} range interactions in order to quantify finite size effects related to the thermodynamic properties of the system. 
The proposed approach relies on a recently proved two-sided Bogoliubov inequality for classical\cite{jstatmecc} and quantum \cite{lmp} systems. The theorem rigorously defines upper and lower bounds to the free energy cost for separating a many-particle system into non-interacting and independent subsystems, and so allows to derive a criterion to assess the accuracy of thermodynamic properties as a function of the system size\cite{prr}. To apply the criterion, it is sufficient to 
calculate the mean particle-particle interaction potential across a dividing surface. Given a suitable reference energy scale, the criterion provides an upper bound for the relative error associated with separating a large system into subsystems that can be simulated independently (at lower computational cost). 

In many cases, evaluating the relative error boils down to a well-behaved numerical integration (i.e.~quadrature or cubature) problem in dimension 3 or 6 that involves a radial distribution function that is typically available from either experimental data or the relevant literature;  for systems with a specific range of interactions or correlations, systematic approximations that can even lead to analytical formulas. We will consider such a situation here and derive a scaling law for the finite-size effects associated with the separation of a dense quantum gas into two independent subsystems. In addition, we look at the error criterion from a conceptual perspective and relate it to thermodynamic response functions.

The outline of the  article is as follows: In Section \ref{sec:relError}, we review the two-sided Bogoliubov inequality for the interface free energy and explain how to derive a computationally feasible, semi-empirical expression to measure the validity of a simulation set-up.  In Section \ref{sec:scalingLaw}, we derive a simple scaling law for the relative error of a dense quantum gas with {finite-range} interactions, Section \ref{sec:thermodynResponse} is devoted to a discussion of the error bounds for the interdomain energy in relation to certain thermodynamic response functions in a perturbative framework. Conclusions are given in Section \ref{sec:fin}. The article contain two appendices: In Appendix \ref{sec:concavity}  we discuss certain properties of the free energy as a function of the interdomain perturbation parameter, Appendix \ref{sec:matrices} records some matrix identities used in the linear oscillator example of Section \ref{ssec:Gaussians}.

\section{A mathematically sound criterion to estimate thermodynamic accuracy as a function of the system size}\label{sec:relError}

In this Section, we formulate the mathematical set-up and briefly review the basic result that allows to assess the relevance of finite-size effects in a molecular simulation. 
To this end, we consider a system of $N$ interacting particles with coordinates $\bmr_{1},\dotsc,\bmr_{N}\in\Omega\subset\R^{d}$, $d=1,2,3$ where $\Omega$ is some bounded subset. The interactions are governed by a Hamiltonian $H\colon\Omega\to\R$, and we assume that the state of the systems is described by a stationary probability density function $f=f(\bmr_{1},\dotsc,\bmr_{N})$ (or: density matrix in the quantum case) that is of the form $f\propto \exp(-\beta H)$ for some $\beta>0$.  

Here we consider a situation in which the particles can be divided into two non-interacting parts, with $n$ particles confined to a volume $\Omega_1\subset\R^d$ and governed by a Hamiltonian $H_1$, and the remaining $m=N-n$ particles in $\Omega_2\subset\R^3$ with Hamiltonian $H_2$. The state inside the two compartments is assumed to be described by probability density functions $f_{1}(\bmr_{1}, \dotsc, \bmr_{n})$ and  $f_{2}(\bmr_{n+1}, \dotsc, \bmr_{N})$ (or: density matrices in the quantum case) of the form $f_i\propto\exp(-\beta H_i)$, $i=1,2$. 

We want to compare the free energy, $F$, of the full systems with the free energy of the system with Hamiltonian $H$ with the free energy, $F_0$, after the the systems has been divided into two non-interacting subsystems with Hamiltonians $H_{1}$ and $H_{2}$. We set  $H=H_{0}+U$, with $H_0 = H_{1}+H_{2}$ where $U$ is the energy of the  interdomain interaction across the separating interface (e.g.~a small tubular neighborhood of the common boundary $\Omega_1\cap\Omega_2$ of $\Omega_1$ and $\Omega_2$), and we consider the free energy difference associated with adding or removing the interaction term: 

\begin{definition}[Interface free energy]\label{def:DeltaF}
    The free energy cost associated with separating the total system with Hamiltonian $H$ into two independent subsystems with Hamiltonians $H_1$ and $H_2$ is defined as
\begin{equation}\label{deltaF}
        \Delta F := F - F_0 = - \beta^{-1} \log\left(\frac{Z}{Z_0}\right) \ .
    \end{equation}
  with 
\begin{equation}
    Z = \int_{\Omega} \ee^{-\beta H(\bmr)}\,\mathrm{d}\bmr\,,\quad Z_0 = \int_{\Omega} \ee^{-\beta H_0(\bmr)}\,\mathrm{d}\bmr\,.
\end{equation} 
\end{definition}

In Refs. \citenum{jstatmecc,lmp}, the following upper and lower bounds to the free energy cost, $\Delta F$, have been proved that hold for partitioning either a classical or a quantum system of particles into two (or more) non interacting subsystems:
\begin{theorem}[Two-sided Bogoliubov inequality]\label{thm:bogo} It holds 
 $$\mathbf{E}_f[U] \le \Delta F \le \mathbf{E}_{f_{1}f_{2}}[U] \,,$$
 where $\mathbf{E}_f[\cdot]$ or $\mathbf{E}_{f_{1}f_{2}}[\cdot]$ are the expectations with respect to the probability density functions $f$ or $f_1f_2$.  
\end{theorem}

It has been further shown that for two-body potentials the criterion of Theorem \ref{thm:bogo} can be reduced to the calculation of one- and two-particle integrals (in dimension $d$ and $2d$):\cite{prr}
\begin{equation}
  \mathbf{E}_{f}[U]=\rho^{2}\int_{\Omega_{1}}\int_{\Omega_{2}}U(\bmr-\bmr^{'})g(\bmr,\bmr^{'})\rd\bmr\rd\bmr^{'}
\label{erhopp}
\end{equation}
with $\bmr\in\Omega_{1}$ and $\bmr^{'}\in\Omega_{2}$, $g(\bmr,\bmr^{'})$ the two-body, e.g. radial distribution function and $\rho=\frac{N}{\Omega}$, and
\begin{equation}
  \mathbf{E}_{f_{1}f_{2}}[U]=\int_{\Omega_{1}}\int_{\Omega_{2}}\rho_{1}(\bmr)\rho_{2}(\bmr^{'}) U(\bmr-\bmr^{'})\rd\bmr\rd\bmr^{'}
  \label{erhozpp}
\end{equation}
with $\bmr\in\Omega_{1}$ and $\bmr^{'}\in\Omega_{2}$, $\rho_{1}(\bmr)$ and $\rho_{2}(\bmr)$ being the position-dependent average particle density in each domain. {While the bounds in Theorem \ref{thm:bogo} require only mild assumptions to guarantee that the upper and lower bounds are finite, the aforementioned simplifications for two-body potentials rely on further assumptions. In particular, the approach applies to system with a sufficient number of particles so that the condition of isotropic system and the radial distribution function are statistically well defined. Such a condition implies that very dilute systems with a small number of particles (e.g. $N\le 10$), as in recent studies of dilute quantum gases \cite{entropy}, cannot be described by Eqs.~(\ref{erhopp}) and (\ref{erhozpp}). Moreover, the physical validity of the approach is only assured for equilibrium systems, far from a (non-equilibrium) phase transition. As it is proved later on, the creation of an interface is equivalent to stimulating the system with a reversible (i.e.~gradient-type) perturbation, thus one has to make sure that such a perturbation occurs far from the condition of a phase transition which may irreversibly destabilize the system.}

\subsection{Relative free energy difference as a quality measure}

Eqs.~\ref{erhopp} and \ref{erhozpp} can be used to define a criterion to quantify the thermodynamic accuracy for a given system size. In Ref.~\citenum{prr}, the following quality parameter that measures $\Delta F$ relative to some characteristic reference (free) energy of the system has been suggested:
\begin{equation}
  q=\frac{|\Delta F|}{|E_{\rm ref}|},
\end{equation}
where $E_{\rm ref}$ can be the free energy of the decoupled system or the average total energy. Since $\Delta F$ is typically not available, it is customary to replace $q$ by the largest relative error, i.e.  
\begin{equation}
  q_{\max}=\frac{\max\big\{|\mathbf{E}_{f}[U]|,|\mathbf{E}_{f_{1}f_{2}}[U]|\big\}
  }{|E_{\rm ref}|}.
    \label{relcrit}
  \end{equation}
  
The parameter $q_{\max}$ is a convenient and easy-to-compute measure of the thermodynamic accuracy of the model that determines the system size required to represent the bulk properties of a system within an order $\mathcal{O}(q_{\max})$.

In the particular case of an isotropic system with an uniform stationary density, Eq.~\ref{erhozpp} can be simplified even further to yield 
\begin{equation}
  \mathbf{E}_{f_{1}f_{2}}[U]=\rho^{2}\int_{\Omega_{1}}\int_{\Omega_{2}}U(\bmr-\bmr^{'})\rd\bmr\rd\bmr^{'}.
  \label{erhozppuni}
\end{equation}

Assuming that either the radial distribution function, $g(\bmr,\bmr^{'})$ or the uniform particle density $\rho$ are known from experimental data or the relevant literature (as is often the case), evaluating the expectations in Eq.~\ref{erhopp} or  Eq.~\ref{erhozppuni} reduces to a simple numerical quadrature (or: cubature) problem in dimensions $d$ or $2d$ where $d\le 3$. For simple systems it is even possible to systematically approximate the corresponding expectations, which then leads to approximate,  but rather general results (e.g. scaling laws); in our previous work we have applied this approach to the uniform interacting electron gas and derived a  basic finite-size scaling law \cite{prr}. 

Such a criterion is complementary to other criteria that account for finite size effects. The complementarity of our approach lies in the fact that $\mathbf{E}_{f}[U]$ and $\mathbf{E}_{f_{1}f_{2}}[U]$ express a thermodynamic response to the perturbation of building an interface, hence it carries a key information regarding the capability of the decoupled system to respond as it would be expected for the bulk of such a substance.
Another advantage lies in the fact that, even though the upper and lower bounds of the interface free energy have a precise analytical form, it is easily possible to derive semi-empirical scaling laws that can be compared or enriched with available empirical data; even for quantum systems, for which evaluating expectations that involve density matrices may be tedious, the resulting cubature problems are relatively straightforward to implement and solve. 

We will discuss such an example in the next section by deriving a scaling law for {system with finite-range interactions} that model quantum gases. 

\section{Quantum gas with finite range potential}
\label{sec:scalingLaw}

We consider a model of a dense quantum gas with {finite-range} interactions in dimension $d=3$.  The most practical model for analytical calculations is a finite potential well (see e.g. Ref.\citenum{revmodphys}):
\begin{equation}
V_{\rm eff}(r)=
  \begin{cases}
  	-V_{0}\,,       & \quad r \le R_{0}\\
  	0\,,  & \quad r > R_{0}
  \end{cases}
\label{potatt}
\end{equation}
Applying formulas (\ref{erhopp}) and (\ref{erhozpp}) readily yields 
\begin{equation}
  \mathbf{E}_{f}[U]=-\rho^{2}V_{0}\int_{\Omega_{1}}\int_{\Omega_{2}}\mathbf{1}_{\{|\bmr-\bmr^{'}|\le R_{0}\}}\,g(\bmr,\bmr^{'})\,\rd\bmr\rd\bmr^{'}
\label{erhoppultra}
\end{equation}
and
\begin{equation}
 \mathbf{E}_{f_{1}f_{2}}[U]=-\rho^{2}V_{0}\int_{\Omega_{1}}\int_{\Omega_{2}}\mathbf{1}_{\{|\bmr-\bmr^{'}|\le R_{0}\}}\,\rd\bmr\rd\bmr^{'}
\label{erhopzultra}
\end{equation}
The total potential energy of the system can be taken as the reference quantity, $E_{\rm ref}$, that is:
\begin{equation}
 E_{\rm ref}=|\mathbf{E}_{f}[U_{\rm tot}]|=\rho^{2}V_{0}\left|\int_{\Omega}\int_{\Omega}\mathbf{1}_{\{|\bmr-\bmr^{'}|\le R_{0}\}}\,g(\bmr,\bmr^{'})\,\rd\bmr\rd\bmr^{'}\right|.
\label{erhoppultratot}
\end{equation}
Specific data about particle density and radial distribution functions are often available in the literature (see e.g. Ref.\citenum{revmodphys}), thus an accurate determination of the total energy and of $ \mathbf{E}_{f}[U]$ and $\mathbf{E}_{f_{1}f_{2}}[U]$ can be done numerically for each specific choice of physical parameters and size of the system on a three dimensional grid by discretizing the integrals using an appropriate cubature formula. Yet, independently of the availability of the radial distribution function of the system,  the condition that $V(r)$ is {finite-range} as in Eq.~\ref{potatt} allows for estimating an approximation to the quality parameter $q$ in form of a scaling law. Specifically, $q$ can be estimated in terms of the number of effective interdomain interactions in relation to the total number of effective interactions over the whole domain, as explained in the next section.

{In the next paragraph, we will derive a simple scaling law to estimate the interdomain energy and the quality parameter $q$. This scaling law will provide rough, yet computationally cheap estimate of the magnitude of the finite range interactions that provides a general trend regarding the relevance of the size effects. 
It does not  depend on the details of the radial distribution function, in particular it does not rely on any specifics regarding particle-particle interactions or particle exclusion (anti-)symmetry. In order to take care of, for example, fermionic or bosonic properties one must use the explicit formulae (\ref{erhopp})--(\ref{erhozpp}) with the respective radial distribution function, which, being a two-body reduced density matrix, would describe the particle symmetries as well. Clearly, finite size effects in quantum systems may be present even in the absence of particle-particle interactions as, for example, in ideal Fermi gas under confinement\cite{confin}; our approach is beneficial when particle-particle interactions play a dominant role in the physics of the system.}

\subsection{Simple scaling law}

Let us consider $E_{\rm tot}$, that is, the absolute value of the total energy of the $N$ particles in the box of size $L^{3}$. Since we deal with an isotropic system with uniform particle density, the domain of interaction of each particle will be the surrounding sphere of radius $R_{0}$ where $R_{0}$ is the interaction cut-off radius, i.e. $V(r)=0$ for $r > R_{0}$. 

By the uniform particle density assumption and ignoring that particles can overlap, the average number of interactions for each particle is roughly given by $(N-1)\varrho$, where $\varrho=4\pi/3 (R_0/L)^3$ is the ratio between the maximum interaction volume and the box volume. Taking into account the symmetry of the particle-particle interactions (i.e.~removing double counting), it therefore follows that the total number of interactions (called ``bulk interactions'' in the following) is given by:
\begin{equation}
	n^{\rm tot}=\frac{N(N-1)}{2}\frac{4\pi}{3}\left(\frac{R_{0}}{L}\right)^{3} \approx N^2\frac{2\pi}{3}\left(\frac{R_{0}}{L}\right)^{3}.
	\label{nitot}
\end{equation}
Recalling that $\rho=N/L^3$ is the uniform particle density inside the simulation box, the last equation states that the average of the total number of interactions is essentially proportional to $N\rho$.

Next we consider the interface that divides the simulation box into two subsystems. The interface is assumed to be aligned with the $yz$-plane and to have a thickness $R_{0}$ in the $x$-direction where we assume that $R_0\ll L$.  Clearly, for values of $R_{0}$ smaller than the mean particle-particle distance, the region within which particles on one side of the interface can interact with those on the other side has volume $\Gamma=L^{2}R_{0}$. Assuming that the particle density is uniform across the interface, the number of particles in this domain is  $N_{\Gamma}=L^{2}R_{0}\rho=N R_{0}/L$. As a consequence, the average number of interdomain interactions is equal to 
\[
\frac{1}{2}\frac{N_\Gamma (N_\Gamma-1)}{2} \frac{4\pi R_0^3}{3L^2R_0} \approx N_\Gamma^2\frac{\pi }{3}  \left(\frac{R_0}{L}\right)^2 = N^2\frac{\pi}{3}  \left(\frac{R_0}{L}\right)^4\,.
\]
The factor $1/2$ in front of the leftmost expression is owed to the fact that on average only half of particles are on one side of the interface; they interact with the other half of the particles on the other side of the interface, provided that they are within distance of at most $R_{0}$. We set
\[
n_\Gamma:= N^{2}\frac{\pi}{3}\left(\frac{R_{0}}{L}\right)^{4}
\]
and call 
\begin{equation}
  E^{\rm tot}_{i}=\frac{E_{\rm tot}}{n^{\rm tot}}
  \label{nidelta}
\end{equation}
the energy per interaction. Under the above assumptions the interface free energy $\Delta F$ can be estimated by the energy per interaction multiplied by the average number of interdomain interactions, $n_\Gamma$. We call this quantity the  average interdomain energy and denote it by $E_{\rm intd}$. Then 
\begin{equation}
	E_{\rm intd} = E_{\rm tot}\frac{n_\Gamma}{n^{\rm tot}} \,.
	\label{eappr}
\end{equation}
If we let $N,L\to \infty$, such that $N/L^3 =\rho$ is kept fixed, it is plausible to assume that the difference between upper and lower bound in Theorem \ref{thm:bogo} will be negligible compared to the reference energy $E_{\rm ref}$ provided that $R_0$ grows at most linearly in $L$ with $R_0/L\to \delta$ for some sufficiently small $\delta>0$. 
\begin{ass*} We suppose that 
	\[
	\lim_{N\to\infty}\frac{|\mathbf{E}_{f_{1}f_{2}}[U] - \mathbf{E}_{f}[U]|}{E_{\rm ref}} = 0\,.
	\]
\end{ass*}

Setting $E_{\rm ref}=E_{\rm tot}$, an estimate of the quality factor $q\approx\frac{E_{\rm intd}}{E_{\rm ref}}$, that measures the strength of the finite size effect, is then, {in such a case, asymptotically equal to $q_{\max}$ defined by (\ref{relcrit})}, since 
\[
0 = \frac{|\mathbf{E}_{f}[U] - \mathbf{E}_{f}[U]|}{|E_{\rm ref}|} \le \lim_{N\to\infty} \frac{|\Delta F -  \mathbf{E}_{f}[U]|}{|E_{\rm ref}|} \le \lim_{N\to\infty}\frac{|\mathbf{E}_{f_{1}f_{2}}[U] - \mathbf{E}_{f}[U]|}{|E_{\rm ref}|} = 0
\]
and 
\[
0 = \frac{|\mathbf{E}_{f_1f_2}[U] - \mathbf{E}_{f_1f_2}[U]|}{|E_{\rm ref}|} \le \lim_{N\to\infty} \frac{|  \mathbf{E}_{f_1f_2}[U]-\Delta F|}{|E_{\rm ref}|} \le \lim_{N\to\infty}\frac{|\mathbf{E}_{f_{1}f_{2}}[U] - \mathbf{E}_{f}[U]|}{|E_{\rm ref}|} = 0
\]
As a consequence, the relative error of the model satisfies the following scaling law
\begin{equation}
	\lim_{N\to\infty} q_{\max} = \lim_{N\to\infty} \frac{n_\Gamma}{n^{\rm tot}}=\frac{\delta}{2},
	\label{scaling}
\end{equation}
with $\delta\approx R_0/L$ for sufficiently large $N$. Such a scaling law can be a practical tool for a quick estimate of the required minimal size of the system for which the corresponding thermodynamic accuracy is within a given threshold. {The only relevant quantity is the characteristic spatial scale of the potential, while the sign of the potential, for example, does not enter into the game. Overall, the simplicity of the criterion can only provide a rough estimates and trends. The approach may be nevertheless useful beyond the field of computational calculations where usually one starts from small systems' size and then upgrades it, depending on the computational resources available.} In fact, the criterion may be used to estimate \emph{a priori} how small a system could be to assure a bulk response and thus use, if possible, smaller samples than the ones currently used.

In the regime of dilute gases as discussed in Ref.\citenum{revmodphys}, with $\rho=\frac{N}{L^{3}}=\frac{10^{-6}}{R_{o}^{3}}$ (i.e., the average particle-particle distance is much larger than the interaction range), the criterion above is not really needed. In fact few particles in a very large volume are sufficient to describe the essence of the bulk properties.  However for medium and high density gases, e.g., $\frac{N}{L^{3}}=\frac{10^{-\alpha}}{R_{o}^{3}}; \alpha\le 3$, the physical relevance of finite-size effects, measured in terms of Eq.~\ref{scaling},  may be sizable and therefore has to be controlled by the computational set-up; for example, for a density of $\rho=\frac{1}{10^{3}R_{0}^{3}}$ for $N=150$ one has: $q\approx 1\%$, however, at the same density for $N=10$ one has $q\approx 2\%$. This means that regarding the statistical consistency of representing the bulk, a small system is sufficient. 


The extreme case is a gas at high density, e.g., with $\alpha=1$ where for $N=150$ one has only $q\approx 5\%$; thus even $N=150$ can be no more considered fully satisfactory.

In the text above we have discussed the physical representability of bulk properties within a statistical mechanics framework. We will now ask the question whether is it possible to relate the accuracy measure, $q$, to other macroscopic thermodynamic quantities. More specifically, we will shed some light on the interpretation of the interface energy bounds as thermodynamic response functions under the perturbation due to adding or removing an interface.

\section{Interdomain energy as a thermodynamic response function}\label{sec:thermodynResponse}

In this section we will elaborate on our interpretation of the interdomain energy as a thermodynamic response to the perturbation of building a separating interface.

\subsection{Particle insertion: $E_{\rm intd}$ and the chemical potential}

The free energy difference $\Delta F$ in Theorem \ref{thm:bogo} is essentially the change of the free energy of the system as the number of interactions changes when the system is separated in different non-interacting subsystems, since the number of interactions diminishes exactly by the number of interdomain interactions when the interaction potential $U$ is added or removed.

Let us assume that prior to separation, the total number of interactions, $n_{\rm i}$, of the system is a certain function $h=h(N)$ of the number of particles, $N$.  After the system has been separated into non-interacting subsystems, a certain number of interactions are no longer present. Recalling that $n_\Gamma$ denotes the number of interdomain interactions, the number of interactions change, $n_{\rm i}$, is given  by $\Delta n_{\rm i}= n_\Gamma$. 

Now let us consider the variation of the free energy with respect to the variation of number of interactions, i.e. we consider $\frac{\Delta F}{\Delta n_{\rm i}}$. Assuming that we can pass to the limit, 
\begin{equation}
\frac{\Delta F}{\Delta n_{\rm i}}\approx \frac{\partial F}{\partial n_{\rm i}}=\frac{1}{h'(N)}\frac{\partial F}{\partial N}=\frac{1}{h'(N)}\mu\,,    
\end{equation}
with $\mu$ the chemical potential of the system.

On the other hand, if we approximate $\Delta F$ by its upper or lower bound, i.e.~$\Delta F\approx  \mathbf{E}_{f}[U]$ or $\Delta F\approx \mathbf{E}_{f_1f_2}[U]$, then
  \begin{equation}
    \frac{1}{h'(N)}\mu\approx \frac{\mathbf{E}_{f^*}[U]}{n_{\Gamma}},
    \label{appapp}
  \end{equation}
where $\mathbf{E}_{f^*}[U]$ denotes the mean interdomain energy with respect to either $f$ or $f_0=f_1f_2$ or some suitable interpolant between $f$ and $f_0$ that provides an even better (computable) approximation of $\Delta F$. (Note that by Theorem \ref{thm:bogo} and the intermediate value theorem, such an interpolant exists, e.g. $f^*=\lambda^* f_0 + (1-\lambda^*)f$ for some $\lambda^*\in[0,1]$.)
Since $N$ is fixed, Eq.~\ref{appapp} states that the average interdomain energy needed to separate the system in two subsystems is approximately proportional to the chemical potential of the system.

We can say more about the relation between the chemical potential and the average interdomain energy: 
Since the chemical potential can be decomposed into $\mu=\mu_{\rm id}+\mu_{\rm ex}$, where $\mu_{\rm id}$ is the ideal gas part of the chemical potential, that depends only on the density of the system and therefore is the same for the total system and for the separated system, while $\mu_{\rm ex}$  is the excess chemical potential that is directly linked to the interactions in the system, it is the latter that counts here. In other words,  $\mu_{\rm ex}$ is what emerges upon building a separating interface. 

In simulations of liquids or gases, for example, $\mu_{\rm ex}$ can be calculated as the 
average work of adding or removing a particle from the system, i.e.~the energetic response to a perturbation due to insertion or removal of a particle. In this case, the excess chemical potential can be computed using the Widom insertion method, see e.g.~Refs.\citenum{ins1,tuckerman,frenkel}. In our situation, $\mathbf{E}_{f}[U]$ or $\mathbf{E}_{f_1f_2}[U]$ are upper and lower bounds to the cost of adding or removing a surface that divides the system in a two smaller fractions, that is the energetic response to the addition or removal of interdomain interactions and which can be easily computed by taking ergodic averages of the interaction potential.

\subsection{Alchemical transformation: free energy perturbation method}

Another argument in favour of the interpretation of  $\mathbf{E}_{f}[U]$ or $\mathbf{E}_{f_{1}f_{2}}[U]$ as a thermodynamic response can be found by considering Zwanzig's free energy perturbation formula for the free energy difference under an alchemical transformation from $H_0$ and $H=H_1$. Letting $\mu_0$ and $\mu_1$ denote the corresponding Boltzmann distributions, and calling $U_\eps=\eps(H_1-H_0)$, we consider the following homotopy (or: alchemical transformation) between the Hamiltonians $H_0$ and $H_1$:  
\begin{equation}
	H_\eps	= (1-\eps) H_0 + \eps H_1 = H_0 + U_\eps \,,\quad \eps\in[0,1]\,.
	\end{equation}	
The free energy difference $\Delta F(\eps)$ between a system in thermal equilibrium with Hamiltonian $H_\eps$, $\eps\in[0,1]$ and $H_0$ can now be expressed by\cite{zwanz}
\begin{equation}\label{deltaFeps}
\Delta F(\eps) = -\beta^{-1}\log \mathbf{E}_{\mu_0}\!\left[e^{-\beta U_\eps}\right]\,,
\end{equation}
where we use the shorthand $\mathbf{E}_{\mu_\eps}[\cdot]$ to denote the expectation with respect to $\mu_\eps$ for any  $\eps\in[0,1]$, i.e., 
\begin{equation}\label{Emueps}
\mathbf{E}_{\mu_\eps}[g] = \frac{1}{Z_\eps}\int_\Omega g(\bmr) e^{-\beta H_\eps(\bmr)}\,\rd\bmr\,,\quad Z_\eps = \int_\Omega e^{-\beta H_\eps(\bmr)}\mathrm{d}\bmr\,,
\end{equation}
where $g\colon\Omega\to\R$ is any integrable function.  
Note that for $\eps=1$ the free energy difference (\ref{deltaFeps}) agrees with the interface free energy (\ref{deltaF}), if we identify the interaction potential $U$ with $U_1=H_1-H_0$ and bear in mind that $\mu_1$ and $\mu_0$ correspond to the former Boltzmann distributions with densities $f\propto\exp(-\beta H_1)$ and $f_1f_2\propto\exp(-\beta H_0)$. 
Theorem \ref{thm:bogo} now immediately implies that 
\begin{equation}\label{DeltaFbogo}
	\mathbf{E}_{\mu_\eps}[U_\eps] \le \Delta F(\eps) \le \mathbf{E}_{\mu_0}[U_\eps]\,, \quad \eps\in[0,1]\,,
\end{equation}
Moreover, by H\"older's inequality, $\Delta F(\eps)$ is a concave function of $\eps$ (see Lemma \ref{lem1} in Appendix \ref{sec:concavity} for details), which implies the following uniform bound for the free energy difference:

\begin{theorem}\label{thm:fep}
	Let $\Delta F<\infty$ for all $\eps\in[0,1]$. Then 
	\begin{equation}\label{unifDeltaF}
			\min\{\mathbf{E}_{\mu_1}[H_1-H_0],0\} \le \Delta F(\eps) \le \max\{\mathbf{E}_{\mu_0}[H_1-H_0],\,0\}\,, \quad \eps\in[0,1]\,.
	\end{equation}
 Moreover, $\Delta F$ is differentiable in $(0,1)$, with 
 \begin{equation}\label{unifDeltaFprime}
     \mathbf{E}_{\mu_1}[H_1-H_0] \le \Delta F'(\eps) \le \mathbf{E}_{\mu_0}[H_1-H_0]
 \end{equation}
\end{theorem}

\begin{proof}
	The upper bound in (\ref{unifDeltaF}) is a straight consequence of (\ref{DeltaFbogo}). To prove the corresponding lower bound, note that Lemma \ref{lem2} implies that $\Delta F$ is differentiable. Moreover, for a differentiable function, concavity is equivalent to its derivative being decreasing. By Lemma \ref{lem2}, 
	\begin{equation}\label{DeltaFderiv}
		\Delta F'(\eps) = 	\mathbf{E}_{\mu_\eps}\!\left[\frac{\partial U_\eps}{\partial\eps}\right] = 	\mathbf{E}_{\mu_\eps}\!\left[H_1-H_0\right]\,.
	\end{equation}
	The lower bound in (\ref{unifDeltaF}) is now implied by (\ref{DeltaFbogo}) and the fact that $\Delta F'$ is a decreasing function. The inequality  (\ref{unifDeltaFprime}) for the derivative follows analogously. 
\end{proof}

While Theorem \ref{thm:fep} can serve as a first step to estimate free energy differences, more precise estimates may be obtained by combining (\ref{DeltaFbogo}) with subgradient estimates, exploiting the fact that any convex function can be lower bounded by affine functions. Specifically, using that $-\Delta F$ is convex, it follows that 
\begin{equation}
	\Delta F(\eps) \le \Delta F(\eps_0) + \Delta F'(\eps_0)(\eps - \eps_0)\,,\quad \eps,\eps_0\in[0,1]\,,
\end{equation}
which together with	 (\ref{DeltaFderiv}) implies that
\begin{equation}
	\eps 	\mathbf{E}_{\mu_1}(H_1-H_0) \le \Delta F(\eps) \le \eps 	\mathbf{E}_{\mu_0}(H_1-H_0) \,,\quad \eps\in[0,1]\,.
\end{equation}
While the last equation is still only a rough estimate, it can be used to extrapolate or interpolate  free energies profiles in a controlled way in situations in which thermodynamic integration is computationally not feasible. By a similar calculation, one can show that the second derivative, $\Delta F''(\eps)$, exists and is equal to the variance of the interaction potential under the probability distribution $\mu_\eps$, which can be exploited to further refine the bounds for $\Delta F$.

A side aspect of Theorem \ref{thm:fep} is that it shows that the mean interdomain energy is the derivative of the free energy with respect to the perturbation parameter $\eps$, which by definition is a response function with respect to an $\eps$-perturbation of the interaction potential. 

\subsection{An illustrative calculation using Gaussians}\label{ssec:Gaussians}

For the sake of the argument, we assume that $\Omega_1=\R^n$ and $\Omega_2=\R^{m}$, $m=N-n$, are unbounded, and we consider two assemblies of linear oscillators that are weakly coupled through an interface. We let $\eps\in [0,1]$ and define the family of Hamiltonians
\begin{equation}
	H_\eps\colon\R^{n+m}\to\R\,,\;  \bmr \mapsto \frac{1}{2} \bmr^T K_\eps\bmr \,,\quad 	
	K_\eps =\begin{pmatrix}
A & C_\eps \\ C^T_\eps & B
	\end{pmatrix}\in \R^{(n+m)\times (n+m)}
\end{equation} 
where the off-diagonal block 
\begin{equation}
	C_\eps =\begin{pmatrix}
		0 & \ldots & 0 \\ \vdots &  \ddots & \vdots\\ \eps & \ldots & 0 
	\end{pmatrix}\in \R^{n\times m}
\end{equation} 
represents the bilinear coupling between the $n^{\text{th}}$ particle of subsystem 1 and the $1^{\text{st}}$ particle of subsystem 2. We assume that $K_\eps$ is symmetric positive definite (s.p.d.) for any $\eps\in [0,1]$, which requires that both $A$ and $B$ are s.p.d.
Our choice entails pairwise intradomain interactions, such as  
\begin{equation}\label{GaussH0ex}
    H_0(\bmr) = \frac{1}{2}\sum_{i=1}^{n-1} k_i \|\bmr_{i+1} - \bmr_i\|^2 + \frac{1}{2}\sum_{j=n}^{N-1} k_j \|\bmr_{j+1} - \bmr_j\|^2
\end{equation}
for $k_i,k_j\ge 0$. Note that the fact that $H_\eps$ is quadratic and strictly convex for all $\eps\in[0,1]$ implies that $\mu_\eps$ is Gaussian. 
We let $\rho_\eps$ denote the associated Gaussian density, so that, in the notation employed in the previous sections, $\rho_1=f$ and $\rho_0=f_1f_2$ correspond to the Boltzmann distributions of the fully coupled and decoupled systems. Specifically, 
\begin{equation}
	\rho_\eps(\bmr) = \sqrt{\frac{\det K_\eps}{(2\pi)^N}}\exp(-\beta H_\eps(\bmr))
\end{equation}
is the density of a Gaussian with mean zero and covariance $\Sigma_\eps=\beta^{-1}K_\eps^{-1}$. We set $U_\eps= H_\eps - H_0$, which boils down to  
\begin{equation}
	U_\eps 
	= \eps\bmr_n\bmr_{n+1}
\end{equation}
In what follows, we will focus on two aspects:
\begin{enumerate}
	\item We will explicitly compute upper and lower bounds to the interface free energy $\Delta F$.
	\item We discuss the asymptotic behaviour of $q_{\text{ref}}$ in the limit $n,m\to\infty$. 
\end{enumerate}

To address the first point, we compute the expectation of the interaction potential $U_\eps$ under the Gaussian equilibrium distribution $\mu_\eps$. As is shown in Appendix \ref{sec:matrices}, the upper bound is given by 
\begin{equation}
	\mathbf{E}_{\mu_0} [U] = 0\,,
\end{equation}
since $\bmr_n$ and $\bmr_{n+1}$ are uncorrelated, whereas the lower bound is strictly negative and thus expresses an anticorrelation between $\bmr_n$ and $\bmr_{n+1}$: 
\begin{equation}
	\mathbf{E}_{\mu_\eps} [U] = -\eps^2a^\sharp_{nn}b^\sharp_{11} + \mathcal{O}(\eps^4)\,.
\end{equation}
Here $a^\sharp_{ij}$ and $b^\sharp_{ij}$ denote the entries of the inverse matrices $A^{-1}$ and $B^{-1}$. 
As a consequence, assuming $\eps$ sufficiently small and ignoring the $\mathcal{O}(\eps^4)$ term, we have 
\begin{equation}\label{DeltaFGaussbound}
	-\eps^2a^\sharp_{nn}b^\sharp_{11} \lesssim \Delta F(\eps) \le 0\,.
\end{equation}

This easily computable and physically intuitive result should be compared with the exact interface free energy (see Appendix \ref{sec:matrices}): 
\begin{equation}\label{DeltaFGauss}
		\Delta F(\eps) = (2\beta)^{-1}\log\frac{\det B_S}{\det B}\,,
\end{equation}
with the Schur complement of $B$, 
\begin{equation}
	B_S = B - C_\eps^TA^{-1}C_\eps =  \begin{pmatrix}
		b_{11} - \eps^2 a^\sharp_{nn} & \ldots & b_{1m} \\ \vdots &  \ddots & \vdots\\ b_{m1} & \ldots & b_{mm} 
	\end{pmatrix},
\end{equation}
being a rank-1 perturbation of $B$ of order $\eps^2$. 

\begin{figure}
    \centering
    \includegraphics[width=0.5\textwidth]{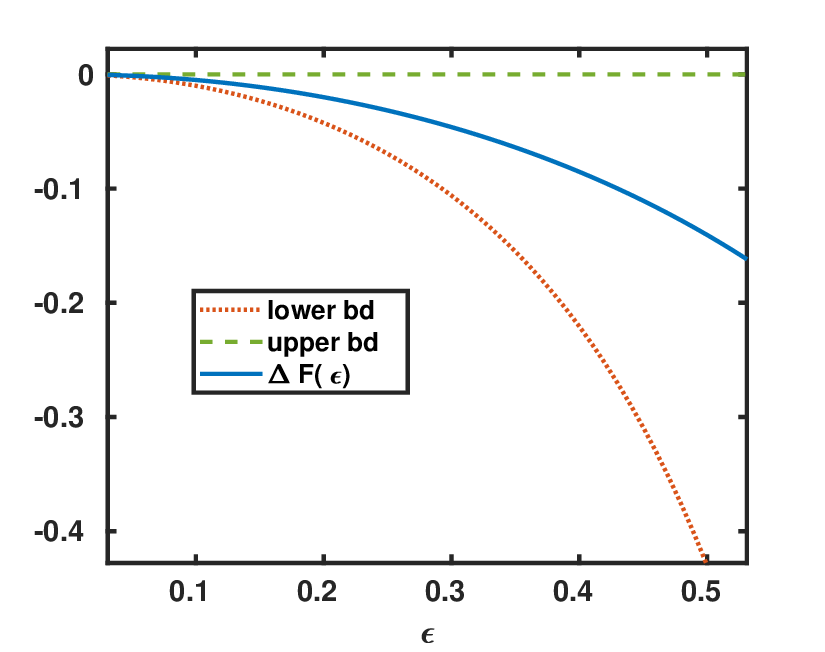}
    \caption{Interface free energy and its upper and lower bounds for two weakly coupled linear oscillator chains, each of which consists of $n=m=100$ particles. Here $\eps$ is the coupling parameter that models the strength of the bilinear interactions between the 100th particle and the 101st particle.}
    \label{fig}
\end{figure}

Figure \ref{fig} shows $\Delta F(\eps)$ for $\beta=1$ together with its upper and lower bounds for two homogeneous linear oscillator chains of $n=m=100$ particles with nearest neighbor interactions, such that $H_0$ is of the form  (\ref{GaussH0ex}) with $k_i=k_j=1$. In this case,   
\begin{equation}
	A = B = \begin{pmatrix}
		2 & -1 & 0 & \ldots & 0 \\ -1 &  2  & -1 & \ldots & 0\\ 0 & \ddots & \ddots & \ddots & 0\\ \vdots & \vdots & \ddots & \ddots  & -1 \\ 0 & \ldots & 0 & -1 & 2  
	\end{pmatrix}\in\R^{n\times n}\,,
\end{equation}
is the negative discrete Laplacians equipped with Dirichlet boundary conditions. In contrast to the upper and lower bounds for $\Delta F$ that can be shown to be a sufficiently accurate approximation in many cases, expressions like (\ref{DeltaFGauss}) are in general not easy to evaluate exactly (even though there is a simple formula here); using asymptotic approximations of $\Delta F(\eps)$ for small $\eps$ is often no better than using the available upper and lower bounds that are asymptotically sharp as $\eps\to 0$.


Before we conclude this section, a final remark is in order. In Section \ref{sec:relError} we have discussed a criterion to check whether finite-size effects play a role in a molecular simulation. Considering our Gaussian example for simplicity, we will argue that while it is crucial to look at the relative error $q_{\max}$ defined in (\ref{relcrit}) rather than just the smallness of the interface free energy, the relative error is relatively robust with respect to the choice of the reference energy scale $E_{\rm ref}$.

As a reference energy, we define the free energy of the decoupled system, $E_{\rm ref}=-\beta^{-1}\log Z_0$. (Note that $-\beta^{-1}\log Z_0 > -\beta^{-1}\log Z_\eps$, since $\Delta F(\eps)<0$ for all $\eps\in(0,1]$, but since we assume $\eps$ to be sufficiently small, the absolute difference is negligible.) Whether or not $q_{\max}$ is small depends on the the interactions within the two compartments. For example, if in the linear oscillator chain case, with $A=B$ being negative discrete Laplacians, Theorem 1.1 of Ref.~\citenum{izyurov2022asymptotics} implies that 
\begin{equation}
	\lim_{N\to\infty} \frac{1}{N}\log \det K_0 = c
\end{equation}
for some constant $c>0$ where $N=2n$. As a consequence
\begin{equation}
	q_{\max} = \frac{|\mathbf{E}_{\mu_\eps}[U]|}{E_{\rm ref}} = \mathcal{O}\left(\frac{\eps^2}{N}\right)\,,
\end{equation}
which gives a rational criterion to check whether a simulation is large enough to ignore finite-size effects (assuming that the corresponding prefactors can be estimated). 

We emphasize that the denominator in $q_{\max}$ may display a similar behaviour even if the interactions inside the compartments are weak (or absent); the reason is that $E_{\rm ref}$ is an extensive quantity. As an extreme  example, we consider the case $A=B=I_{n\times n}$, which gives 
\begin{equation}
	E_{\rm ref} = -\beta^{-1}\log Z_0 = (2\beta)^{-1}\log \frac{\det A\det B}{(2\pi)^N} = \mathcal{O}(N) \,.
\end{equation}
The same scaling behaviour is achieved when the average potential energy instead of the free energy is considered (which amounts to computing the trace of the covariance matrix rather than the logarithm of its determinant), and it is merely a matter of taste or computational convenience which reference energy scale is chosen; the bottomline is that the relative error criterion is expected to be relatively robust with respect to the choice of the reference energy. Nevertheless, when it comes to tuning the system size, the prefactor in the $\mathcal{O}(N)$-term becomes important, because it determines how large $N$ or $n$ must be chosen to achieve the desired thermodynamic accuracy exactly.

\section{Conclusions}\label{sec:fin}



 We have derived a scaling law for the relative error of certain thermodynamic bulk properties induced by finite-size effects in a dense quantum gas. In previous works, the relative error criterion was applied to the simulation of thermodynamic bulk properties of many-particle systems, including
an interacting uniform electron gas, as often used in Quantum Monte Carlo simulations of cold quantum gases. 

The approach is based on a two-sided Bogoliubov inequality for the interface free energy that measures the free energy difference when an interaction potential for the particles in two otherwise non-interacting subsystems is added. We stress that there is a trade-off between an accurate estimate of the interface free energy in terms of its upper and lower bounds (that may be computationally costly to evaluate) and a scaling law that can be applied without doing any (or very little) numerical computations, but that gives only the correct order of magnitude of the required system size; when it comes to optimizing the size of the simulation box (i.e. the number of particles), the prefactors clearly matter. 

Finally, in order to understand how the relative error is related to thermodynamic bulk properties, we have discussed the interpretation of the free energy error bounds as approximate thermodynamic response functions. We admit that this part is merely conceptual, and therefore future work ought to address, how finite-size effects affect the relevant thermodynamic response functions, such as isothermic or isobaric compressibilities. 

\acknowledgments{This work was partly supported by the DFG Collaborative Research Center 1114 ``Scaling Cascades in Complex Systems'', project No.235221301, Projects A05 (C.H.) ``Probing scales in equilibrated systems by optimal nonequilibrium forcing'' and C01 (L.D.S.) ``Adaptive coupling of scales in molecular dynamics and beyond to fluid dynamics''. 
One author (C.H.) acknowledges support by the German Federal Government, the Federal Ministry of Education and Research and the State of Brandenburg within the framework of the joint project ``EIZ: Energy Innovation Center'' (project numbers 85056897 and 03SF0693A) with funds from the Structural Development Act (Strukturstärkungsgesetz) for coal-mining regions.}

\appendix

\section{Concavity and differentiability of the interface free energy}\label{sec:concavity}

In this paragraph, we collect a few properties of the free energy difference under an alchemical transformation between the Hamiltonians $H_0$ and $H_1$,
\begin{equation}
\Delta F(\eps) = -\beta^{-1}\log \mathbf{E}_{\mu_0}\!\left[e^{-\beta U_\eps}\right]\,,\quad \eps\in[0,1]\,,
\end{equation}
where $U_\eps = \eps(H_1 - H_0)$ and we use the shorthand $\mathbf{E}_{\mu_0}[\cdot]$ to denote the expectation with respect to the measure $\mu_0$ with density
$\rho_0\propto\exp(-\beta H_0)$; since $\mu_0$ is independent of $\eps$, we will drop the subscript $\mu_0$ in the following. 

The properties of the free energy difference are standard (see, e.g., Sec.~2 in Ref.~\citenum{dembo2009large}) and essentially follow from the properties of its close relative, the cumulant generating function of a random variable (also called: logarithmic moment generating functions). For the ease of notation, we switch from free energies---that involve parameters $\beta$ and the like---to cumulant generating functions.
\begin{definition}[Cumulant generating function]
    Letting $W$ be any real valued random variable (i.e.~observable) on $\Omega$, its cumulant generating function (CGF) is defined as
\begin{equation}
    \gamma\colon \R\to(-\infty,\infty]\,,\; s \mapsto \log \mathbf{E}\!\left[e^{s W}\right]=:\gamma(s)\,,
\end{equation}
where the value $\gamma(s)=\infty$ is allowed. 
\end{definition}
If we set $W=-\beta (H_1-H_0)$, then the free energy is recovered by $\Delta F(\eps) = -\beta^{-1}\gamma(\eps)$, $\eps\in[0,1]$. In particular, $\Delta F$ is concave if and only if $-\Delta F$ is convex which is exactly the case if $\gamma$ is convex on the domain $[0,1]$. 

\begin{lemma}\label{lem1}
    The function $\gamma$ is convex. 
\end{lemma}

\begin{proof}
    We remind the reader of Hölder's inequality: For any two random variables $X,Y$ on $\Omega$ and $q,p\in [1,\infty]$ with $\frac{1}{p} + \frac{1}{q}=1$, it holds that $\mathbf{E}[|XY|]\le (\mathbf{E}[|X|^p])^{\frac{1}{p}}(\mathbf{E}[|Y|^q])^{\frac{1}{q}}$. 

    Now letting $s,t\in\R$ and $\lambda\in[0,1]$, then Hölders inequality for $p=\frac{1}{\lambda}$ and $q=\frac{1}{1-\lambda}$ with the convention $\frac{1}{0}=\infty$ and $\frac{1}{\infty}=0$ implies that 
        \begin{align*}
        \gamma(\lambda s + (1-\lambda)t) & = \log\mathbf{E}\!\left[e^{\lambda s W}e^{(1-\lambda)t W}\right]\\
        & = \log\left(\mathbf{E}\!\left[\left(e^{s W}\right)^\lambda\left(e^{t W}\right)^{1-\lambda}\right]\right)\\
        & \le \log\left(\left(\mathbf{E}\!\left[e^{s W}\right]\right)^\lambda\left(\mathbf{E}\!\left[e^{t W}\right]\right)^{1-\lambda}\right)\\
        & = \lambda \gamma(s) + (1-\lambda)\gamma(t)\,.
        \end{align*}
    Hence $\gamma$ is convex. 
\end{proof}

As a consequence of Lemma \ref{lem1}, the interface free energy $\Delta F(\eps)$ is a concave function of $\eps$. The next lemma addresses the differentiability of CGF and free energy. 

\begin{lemma}\label{lem2}
    Let $D:=\{s\in\R\colon \gamma(s)<\infty\}$. Then $\gamma$ is differentiable in the interior of $D$, with 
    \begin{equation}
        \gamma'(s) = \frac{\mathbf{E}\!\left[We^{sW}\right]}{\mathbf{E}\!\left[e^{sW}\right]} = \mathbf{E}\!\left[We^{sW - \gamma(s)}\right] =: \mathbf{E}_{\mu_s}[W]\,,\quad s\in D^\circ
    \end{equation}
\end{lemma}

\begin{proof}
It suffices to show that the moment generating function $\varphi(s)=\mathbf{E}[e^{sW}]$ is differentiable in $D^\circ$. 
Clearly $0\in D$, so $D \neq \emptyset$. We further assume that its interior, $D^\circ$, in non-empty. It is easy to see that $D$ must be an interval, since $\delta\in D$ for some $\delta>0$ implies that $[0,\delta]\in D$. To see the latter, notice that for any $s\in[0,\delta]$, the following inequality holds: 
\[
    \mathbf{E}[e^{s W}] =  \mathbf{E}[e^{s W}\mathbf{1}_{\{W \ge 0\}}]   + \mathbf{E}[e^{s W}\mathbf{1}_{\{W<0\}}]  \le \mathbf{E}[e^{\delta W}] + 1\,.
\]
By the same argument, $-\delta\in D$ for some $\delta>0$ implies that $[-\delta,0]\in D$. 

Another property of the moment-generating function that is finite on an interval concerns the existence of moments; here we are only interested in the mean: If $s\in D^\circ$ is any interior point, then   
\[
    \mathbf{E}\!\left[|W|e^{sW}\right] < \infty \,.
\]
The assertion follows from choosing $\delta>0$ sufficiently small such that $(s-\delta,s+\delta)\subset D$, because 
\begin{equation}\label{aux1}
|W|e^{sW} \le e^{\delta|W|}e^{sW} \le e^{(s-\delta)W} + e^{(s+\delta)W}\,,
\end{equation}
where the right hand is integrable.  

Now, to prove that $\varphi$ is differentiable in $D^\circ$, we consider $s,t\in D^\circ$ and note that 
\begin{equation}\label{aux2}
\frac{\varphi(t) - \varphi(s)}{t-s} - \mathbf{E}\!\left[We^{sW}\right] = \mathbf{E}\!\left[\frac{e^{tW} - e^{sW}}{t-s}  - We^{sW}\right] = \mathbf{E}\!\left[e^{sW}\left(\frac{e^{(t-s)W} - 1 - (t-s)W}{t-s}\right)\right]    
\end{equation}
If $|t-s|<\delta$, then 
\begin{align*}
\left|e^{sW}\left(\frac{e^{(t-s)W} - 1 - (t-s)W}{t-s}\right)\right| & = \left|e^{sW}\sum_{k\ge 2} \frac{(t-s)^{k-1} W^k}{k!}\right|\\
& \le |W|e^{sW}\sum_{k\ge 2} \frac{|\delta W|^{k-1}}{(k-1)!}\\
& \le |W|\left(e^{(s-\delta)W} + e^{(s+\delta)W}\right)\,.
\end{align*}
By iterating (\ref{aux1}), bearing in mind that $s\pm\delta\in D^\circ$, the last expression is integrable, so Lebesgue's Dominated Convergence Theorem applied to (\ref{aux2}) yields the desired result: 
\[
\lim_{t\to s} \left|\frac{\varphi(t) - \varphi(s)}{t-s} - \mathbf{E}\!\left[We^{sW}\right] \right|= 0
\] 
As a consequence, $\gamma'(s)=\frac{d}{ds}\log\varphi(s)=\mathbf{E}\!\left[We^{sW - \gamma(s)}\right]$.
\end{proof}

Lemma \ref{lem2} can be readily translated to the free energy framework: If we set $W=-\beta (H_1-H_0)$, then $\Delta F(\eps) = -\beta^{-1}\log\varphi(\eps)$. Assuming that $\Delta F(\eps)$ is finite for all $\eps\in[0,1]$, it follows that
\begin{equation}
    \Delta F'(\eps) = \mathbf{E}_{\mu_\eps}[H_1-H_0]\,.
\end{equation}

\section{Matrix calculations}\label{sec:matrices}

We explain the calculations that lead to lead to the result (\ref{DeltaFGaussbound}) in the linear oscillator example of Section \ref{ssec:Gaussians}. To this end recall, that $\mu_\eps=\mathcal{N}(0,\Sigma_\eps)$ with $\Sigma_\eps=\beta^{-1}K_\eps^{-1}$. To compute the inverse of the stiffness matrix $K_\eps$, we define the Schur complements of the two matrices $A$ and $B$: 
\begin{equation}
	A_S = A - C_\eps B^{-1} C_\eps^T\,,\quad 	B_S = B - C_\eps^T A^{-1} C_\eps\,,\quad 
\end{equation}
By the properties of Schur complements, both $A_S$ and $B_S$ are s.p.d., hence invertible, which allows us to express the covariance matrix in the form 
\begin{equation}
\Sigma_\eps = \beta^{-1}\begin{pmatrix}
	A_S^{-1} & - A^{-1}C_\eps B_S^{-1}\\ - B_S^{-1}C_\eps^T A^{-1} & B_S^{-1}\,.
\end{pmatrix}
\end{equation}
Since $C_\eps$ is a rank-1 matrix, so is the off diagonal block. To compute it, we denote by $a_{ij}$ and $b_{ij}$ the entries of the matrices $A$ and $B$, and by $a^\sharp_{ij}$ and $b^\sharp_{ij}$ the corresponding entries of $A^{-1}$ and $B^{-1}$. Since
\begin{equation}
	C_\eps^T A^{-1}C_\eps = \begin{pmatrix}
		\eps^2 a^\sharp_{nn} & \ldots & 0 \\ \vdots &  \ddots & \vdots\\ 0 & \ldots & 0 
	\end{pmatrix}\,,\quad  C_\eps B^{-1}C_\eps^T = \begin{pmatrix}
	0 & \ldots & 0 \\ \vdots &  \ddots & \vdots\\ 0 & \ldots & \eps^2 b^\sharp_{11} 
	\end{pmatrix}\,,
\end{equation}
we can recast the Schur complements of $A$ and $B$ as
\begin{equation}
	A_S = A - uu^T\,,\quad u = \left(0,\ldots,0,\eps\sqrt{b^\sharp_{11}}\right)^T
\end{equation}
and 
\begin{equation}\label{BS}
	B_S = B - vv^T\,,\quad v = \left(\eps\sqrt{a^\sharp_{nn}},0,\ldots,0\right)^T\,,
\end{equation}
where the real square roots of $a^\sharp_{nn}$ and $b^\sharp_{11}$ are well-defined since both $A^{-1}$ and $B^{-1}$ are s.p.d.~which implies that their diagonal entries are positive.

To compute the lower bound of the interface free energy, we need to compute the covariance between $\bmr_n$ and $\bmr_{n+1}$ that depends on the off-diagonal entry of $\Sigma_\eps$. The inverse of the Schur complements can be computed using the Woodbury identity for rank-1 perturbations of invertible matrices (a.k.a.~Sherman-Morrison formula, see p.~50 in Ref.\citenum{Golub1996}), which yields 
\begin{equation}
	A_S^{-1} = A^{-1} - \frac{A^{-1}uu^TA^{-1}}{1+ u^T A^{-1}u}\,,\quad 	B_S^{-1} = B^{-1} - \frac{B^{-1}vv^TB^{-1}}{1+ v^T B^{-1}v}\,.
\end{equation}
Putting everything together and ignoring terms of order 4 in $\eps$, we obtain the free energy bound (\ref{DeltaFGaussbound}): 
\begin{equation}
	-\eps^2a^\sharp_{nn}b^\sharp_{11} \lesssim \Delta F(\eps) \le 0\,.
\end{equation}

In order to compute the exact interface free energy, $\Delta F(\eps)$, we have to compute the determinant of the block matrix $K_\eps$ that we factorize as follows:
\begin{equation}
    \begin{pmatrix}
A & C_\eps \\ C^T_\eps & B
	\end{pmatrix} = \begin{pmatrix}
I & 0 \\ C^T_\eps A^{-1} & I
	\end{pmatrix} \begin{pmatrix}
A & 0 \\ 0 & B - C_\eps^T A^{-1} C_\eps
	\end{pmatrix}\begin{pmatrix}
I & A^{-1}C_\eps \\ 0 & I
	\end{pmatrix}
\end{equation}
with $B - C_\eps^T A^{-1} C_\eps$ being equal to the Schur complement $B_S$. Since the two outer triangular matrices in the matrix product on the right hand side have unit determinant, it follows that  
\begin{equation}
	\det K_\eps = \det A\,\det B_S\,.
\end{equation}
Specifically, for $\eps=0$, we have
\begin{equation}
	\det K_0 =  \det A\,\det B\,,
\end{equation}
which yields (\ref{DeltaFGauss}):  
\begin{equation}
	\begin{aligned}
		\Delta F(\eps) & = -\beta^{-1}\log\frac{Z_\eps}{Z_0}\\
				 & = -\beta^{-1}\log\sqrt{\frac{\det K_0}{\det K_\eps}}\\
				 & = (2\beta)^{-1}\log\frac{\det B_S}{\det B}\,,
	\end{aligned}
\end{equation}
where the Schur complement of $B$ can be easily computed using (\ref{BS}):  
\begin{equation}
	B_S =  \begin{pmatrix}
		b_{11} - \eps^2 a^\sharp_{nn} & \ldots & b_{1m} \\ \vdots &  \ddots & \vdots\\ b_{m1} & \ldots & b_{mm}
	\end{pmatrix}\,.
\end{equation}

\bibliography{literature}
\end{document}